\documentclass[english,conference]{IEEEtran}
\usepackage[T1]{fontenc}
\usepackage[latin9]{inputenc}
\usepackage{units}
\usepackage{amsmath}
\usepackage{amssymb}
\usepackage{graphicx}

\makeatletter
\usepackage{amsthm}
\usepackage{comment}
\usepackage{subfigure}
\usepackage{setspace}

\newtheorem{proposition}{Proposition}

\allowdisplaybreaks

\makeatother

\usepackage{babel}
\begin{document}

\title{End-to-End Known-Interference Cancellation (E2E-KIC) with Multi-Hop
Interference}

\author{\IEEEauthorblockN{Shiqiang Wang\IEEEauthorrefmark{1}\IEEEauthorrefmark{2},
Qingyang Song\IEEEauthorrefmark{1}, Kailai Wu\IEEEauthorrefmark{1},
Fanzhao Wang\IEEEauthorrefmark{1}, Lei Guo\IEEEauthorrefmark{1}}\IEEEauthorblockA{\IEEEauthorrefmark{1}School
of Computer Science and Engnineering, Northeastern University, Shenyang,
China} \IEEEauthorblockA{\IEEEauthorrefmark{2}Department of
Electrical and Electronic Engineering, Imperial College London, United
Kingdom} \IEEEauthorblockA{Email: shiqiang.wang@ieee.org, songqingyang@mail.neu.edu.cn,
kailaiwu@gmail.com,\\ fanzhaowang@gmail.com, guolei@cse.neu.edu.cn}}
\maketitle
\begin{abstract}
Recently, end-to-end known-interference cancellation (E2E-KIC) has
been proposed as a promising technique for wireless networks. It sequentially
cancels out the known interferences at each node so that wireless
multi-hop transmission can achieve a similar throughput as single-hop
transmission. Existing work on E2E-KIC assumed that the interference
of a transmitter to those nodes outside the transmitter's communication
range is negligible. In practice, however, this assumption is not
always valid. There are many cases where a transmitter causes notable
interference to nodes beyond its communication distance. From a wireless
networking perspective, such interference is caused by a node to other
nodes that are multiple hops away, thus we call it \emph{multi-hop
interference}. The presence of multi-hop interference poses fundamental
challenges to E2E-KIC, where known procedures cannot be directly applied.
In this paper, we propose an E2E-KIC approach which allows the existence
of multi-hop interference. In particular, we present a method that
iteratively cancels \emph{unknown }interferences transmitted across
multiple hops. We analyze mathematically why and when this approach
works and the performance of it. The analysis is further supported
by numerical results. \end{abstract}


\section{Introduction}

Wireless multi-hop networks have attracted extensive interest due
to its flexibility of deployment and adaptability to changing network
topologies. These features are particularly useful for many emerging
applications today, such as vehicular networking. One big problem
in multi-hop networks is its performance degradation compared to single-hop
networks. This is illustrated in Fig. \ref{fig:basicConcept}(a),
where node $1$ can only transmit one packet every three timeslots,
due to the half-duplex nature of conventional radio transceivers and
to avoid collision of adjacent packets. Obviously, if the transmission
between the source and destination nodes are carried out in a single
hop, node $1$ can transmit one packet \emph{per }timeslot\footnote{In a large network, there may exist other transmitting nodes for which
collision needs to be avoided, but we ignore that case here for simplicity.
We also ignore packet losses, congestion, etc.} to node $4$. It turns out that in this simple example, multi-hop
transmission brings only one third of the throughput of single-hop
transmission. A natural question is: Can we close this gap?

\begin{figure}
\center{\includegraphics[width=0.95\linewidth]{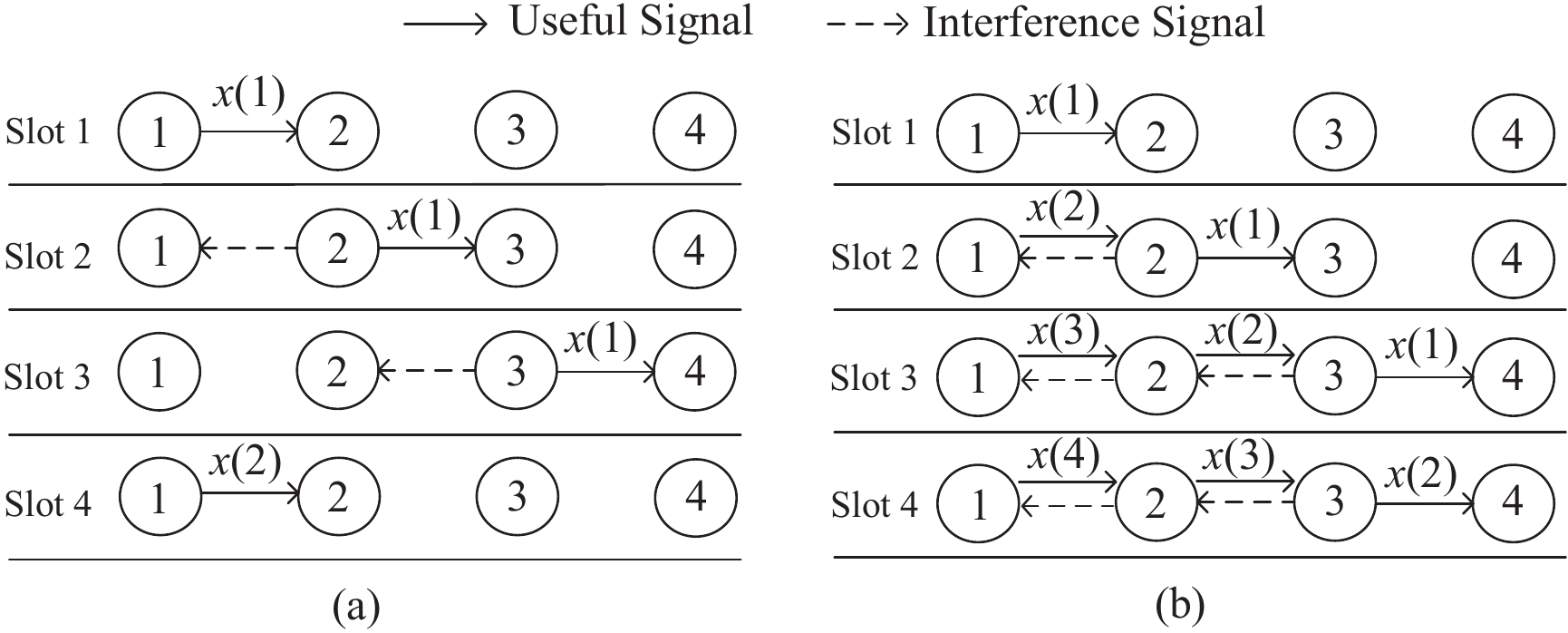}}

\caption{Procedure of conventional multi-hop communication (a) and E2E-KIC
(b), where node $1$ sends packets to node $4$ via nodes $2$ and
$3$.}
\label{fig:basicConcept}
\end{figure}

Efforts have been made to improve the performance of multi-hop networks
with advanced physical-layer techniques. These techniques are generally
based on the idea that interference signals known to the receiver
a priori can be cancelled out from a superposed signal, where the
superposed signal contains an unknown signal carrying new data and
a set of known interferences. Such techniques are termed as known-interference
cancellation (KIC) \cite{BKIC}. 

One representative KIC technique is physical-layer network coding
(PNC), where the relay node encodes the superposed signal into a new
signal and broadcasts it out to nodes with decoding capability. The
typical application of PNC is in two-way relay networks, which have
two end nodes exchanging data via a common relay \cite{PNCTwoWay}.
In this setting, PNC is able to transfer two packets every two slots
for two symmetric data flows\footnote{We refer to bidirectional data flows as two separate unidirectional
data flows in this paper.} in opposite directions. On average, this corresponds to a throughput
of $0.5\textrm{ packet/(slot}\cdot\textrm{flow)}$. While PNC can
be extended to linear networks with more than two hops \cite{refHotTopic},
its throughput cannot exceed $0.5\textrm{ packet/(slot}\cdot\textrm{flow)}$.

Another representative KIC technique is full duplex (FD) communication
\cite{fullDuplex,sundaresan2014full}. Here, wireless nodes can transmit
and receive at the same time, by effectively cancelling its transmitted
signal at the receiving antenna in real time. FD can achieve a throughput
of $1\textrm{ packet/(slot}\cdot\textrm{flow)}$ for bidirectional
single-hop communications (containing two opposite flows), which is
the same as what one can achieve for unidirectional single-hop communications.
However, it is non-straightforward to directly extend FD to multi-hop
communications.

End-to-end KIC (E2E-KIC) was recently proposed to boost up the throughput
of wireless multi-hop networks. Unidirectional E2E-KIC supporting
a single flow can be achieved by equipping nodes with FD transceivers
\cite{SE2EKICMAC}; bidirectional E2E-KIC supporting two opposite
flows can be achieved by leveraging both PNC and FD techniques \cite{BE2E-KIC}.
It was shown that both unidirectional and bidirectional E2E-KIC can
achieve a throughput of $1\textrm{ packet/(slot}\cdot\textrm{flow)}$
when ignoring a few idle timeslots caused by propagation delay. This
implies that \emph{the throughput of E2E-KIC is (almost) the same
as the single-hop counterparts}. 

The basic concept of E2E-KIC\footnote{We only focus on unidirectional E2E-KIC in this paper, and refer to
it as E2E-KIC.} is shown in Fig. \ref{fig:basicConcept}(b). The source node $1$
starts the transmission by sending out its first packet $x(1)$ to
node $2$ in slot $1$. Then, in slot $2$, node $1$ sends its second
packet $x(2)$ and node $2$ relays the first packet $x(1)$ to node
$3$. Node $2$ can receive $x(2)$ while transmitting $x(1)$ due
to its FD capability. In slot $3$, node $1$ sends $x(3)$, node
$2$ sends $x(2)$, and node $3$ sends $x(1)$. Now, node $2$ can
receive $x(3)$ because it can cancel out $x(2)$ transmitted by itself
as well as $x(1)$ transmitted by node $3$. It can cancel out $x(1)$
because it has received $x(1)$ before and knows its contents. Node~$3$
can receive $x(2)$ due to its FD capability. Slot $4$ and all remaining
slots are then analogous. With this approach, node $1$ is able to
transmit one packet per slot, and node $4$ is able to receive one
packet per slot starting from slot~$3$.

In this example as well as existing work \cite{SE2EKICMAC,BE2E-KIC},
multi-hop interference has not been considered, i.e., it has been
assumed that the signal transmitted by node $1$ is negligible at
node $3$, for instance. If this is not the case, the existing E2E-KIC
scheme becomes inapplicable. Consider slot $2$ in Fig. \ref{fig:basicConcept}(b)
as an example, if the signal transmitted by node $1$ causes strong
interference at node $3$, node $3$ may not be able to correctly
receive $x(1)$ from node~$2$. Unfortunately, such non-negligible
\emph{multi-hop interference }widely exist in practical scenarios.
It is quite possible that two nodes are not close enough to correctly
receive the transmitted data, but also not far away enough to ignore
the transmitter's interference signal \cite{interferenceRange}. Therefore,
we need to investigate whether and how E2E-KIC works in scenarios
with multi-hop interference. In this paper, we propose an E2E-KIC
scheme that is feasible for such scenarios.

\section{System Model \label{sec:System-Model}}

Similar to related work \cite{refHotTopic,SE2EKICMAC,BE2E-KIC}, we
focus on a chain topology with $N$ nodes. We note that these nodes
can be part of a larger network with arbitrary topology, and the chain
sub-network can be obtained via MAC protocols that support E2E-KIC
\cite{SE2EKICMAC}. Nodes in the chain topology are sequentially indexed
with $1,2,...,N$, where node~$1$ is the source node, node~$N$
is the destination node. We use $t$ to denote the timeslot index
starting at $1$, and use $x(t)$ to denote the signal (containing
data packet\footnote{We use $x(t)$ to denote both the data packet and its corresponding
signal in this paper.}) that the \emph{source node} (indexed $1$) sends out in timeslot
$t$. Let $h_{ji}$ denote the channel coefficient from node $j$
to node $i$. We consider frequency-flat, slow-varying channels in
this paper, so that $h_{ji}$ does not vary over our time of interest.
We also assume that $h_{ji}$ is known a priori, while anticipating
that techniques similar to blind KIC \cite{BKIC} or correlation-based
approaches \cite{zigzagDecoding} may be applied for cases with unknown
$h_{ji}$.

\begin{figure}
\center{\includegraphics[width=0.7\linewidth]{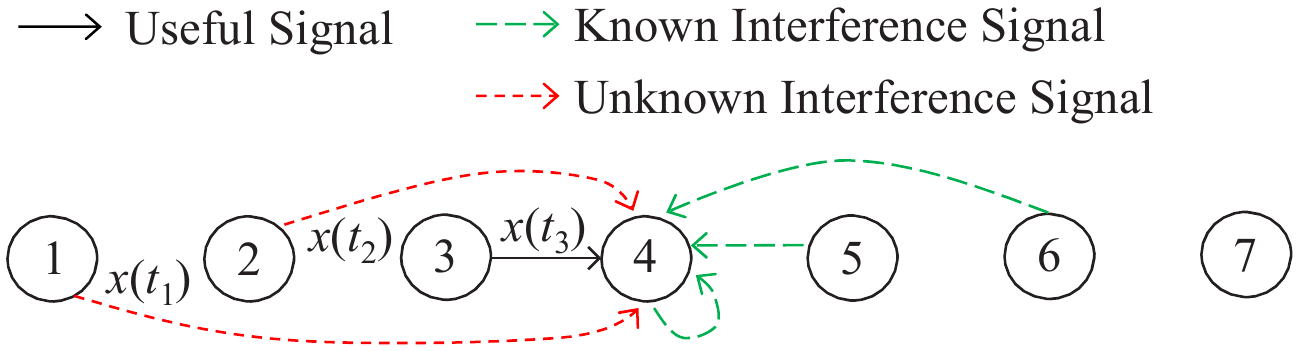}}

\caption{Useful and (multi-hop) interference signals at node $4$.}


\label{fig:multihopInteference}
\end{figure}

We assume that when a node receives a superposed signal, it can successfully
cancel out all its known signals. For example (see Fig. \ref{fig:multihopInteference}),
when $N=7$, node $4$ knows the signals sent by nodes $5$ and $6$
because they were previously relayed by node $4$. Thus, node~$4$
can cancel out the signals transmitted by nodes~$5$ and $6$ as
well as its own transmitted signal (due to FD capability). Node~$4$
intends to receive the signal sent by node~$3$, while nodes~$1$
and $2$ are potential multi-hop interferers transmitting\emph{ unknown
signals }for node~$4$. 

For an arbitrary node $i$ (node $4$ in the above example), after
cancelling out all known signals from nodes $i$ (itself), $i+1,...,N$,
the remaining received signal in slot $t$ is
\begin{equation}
y_{i}(t)=\sum_{j=1}^{i-1}h_{ji}\cdot x\left(t_{j}\right)+z_{i}(t),\label{eq:systemModel}
\end{equation}
where $x\left(t_{i-1}\right)$ is the packet that node $i$ intends
to receive (which is sent by node $i-1$), $\sum_{j=1}^{i-2}h_{ji}\cdot x\left(t_{j}\right)$
is the sum of unknown interference signals sent by nodes $1,...,i-2$,
and $z_{i}(t)$ is the noise signal with power $\sigma^{2}$. Residual
interference caused by imperfect KIC can be considered as part of
the noise, and we do not specifically study it in this paper. We use
$P_{T}$ to denote the transmission power so that $\mathrm{E}\left[|x(t)|^{2}\right]=P_{T}$. 

The variables $t_{j}$ in (\ref{eq:systemModel}) stand for the slot
index when the packet has been sent out by node~$1$ \emph{for the
first time} (see definition of $x(t)$ earlier). We always have $t=t_{1}>t_{2}>...>t_{i-2}>t_{i-1}$
due to the packet relaying sequence. The specific values of $t_{j}$
are related to the interference strength. For single-hop interference
(see Fig. \ref{fig:basicConcept}(b)), we always have $t_{j}-t_{j-1}=1$.
For multi-hop interference, it is possible that $t_{j}-t_{j-1}>1$
(see next section), which means that we may need to wait for more
than one slot before we can cancel out all interferences and successfully
decode the packet. We define $\Delta_{j}=t-t_{j}$ to denote the total
delay incurred for transmitting the packet from node~$1$ to node~$j$.

For simplicity, we assume that packets can always be correctly received
if the signal-to-interference-plus-noise ratio (SINR) is larger than
a threshold $\gamma$. In practice, this can be achieved by a properly
designed error-correcting code. For the convenience of analysis later
in this paper, we assume $\gamma\geq1$.

\section{Interference Cancellation Scheme}

We focus on how to cancel\footnote{Strictly speaking, instead of cancelling, we aim to reduce the strength
of unknown interferences to a satisfactory level so that the packet
can be correctly decoded. For simplicity, we use the term ``cancel''
in this paper.} the \emph{unknown} multi-hop interferences, i.e., the $\sum_{j=1}^{i-2}h_{ji}\cdot x\left(t_{j}\right)$
terms in (\ref{eq:systemModel}). Cancellation of known interferences
is analogous with existing KIC work, so we do not consider it here.

\begin{figure*}
\begin{align}
g_{i,m}(t) & \triangleq g_{i,m-1}(t)\pm\sum_{j_{1}=1}^{i-2}\sum_{j_{2}=1}^{i-2}\cdots\sum_{j_{m-1}=1}^{i-2}\sum_{j_{m}=1}^{i-2}\frac{h_{j_{1}i}h_{j_{2}i}\cdots h_{j_{m}i}}{h_{(i-1),i}^{m}}y_{i}(t+\delta_{j_{1}}+\delta_{j_{2}}+\ldots+\delta_{j_{m}})\label{eq:mthIteration1}\\
 & =h_{(i-1),i}x(t_{i-1})\pm\sum_{j_{1}=1}^{i-2}\sum_{j_{2}=1}^{i-2}\cdots\sum_{j_{m+1}=1}^{i-2}\frac{h_{j_{1}i}h_{j_{2}i}\cdots h_{j_{m+1}i}}{h_{(i-1),i}^{m}}x(t_{i-1}+\delta_{j_{1}}+\delta_{j_{2}}+\ldots+\delta_{j_{m+1}})\nonumber \\
 & \quad\pm\sum_{j_{1}=1}^{i-2}\sum_{j_{2}=1}^{i-2}\cdots\sum_{j_{m}=1}^{i-2}\frac{h_{j_{1}i}h_{j_{2}i}\cdots h_{j_{m}i}}{h_{(i-1),i}^{m}}z_{i}(t+\delta_{j_{1}}+\delta_{j_{2}}+\ldots+\delta_{j_{m}})\mp\ldots\nonumber \\
 & \quad+\sum_{j_{1}=1}^{i-2}\sum_{j_{2}=1}^{i-2}\frac{h_{j_{1}i}h_{j_{2}i}}{h_{(i-1),i}^{2}}z_{i}(t+\delta_{j_{1}}+\delta_{j_{2}})-\sum_{j_{1}=1}^{i-2}\frac{h_{j_{1}i}}{h_{(i-1),i}}z_{i}(t+\delta_{j_{1}})+z_{i}(t)\label{eq:mthIteration2}
\end{align}

\hrulefill
\end{figure*}

\subsection{Constructive Example}

\begin{figure}
\center{\includegraphics[width=0.65\linewidth]{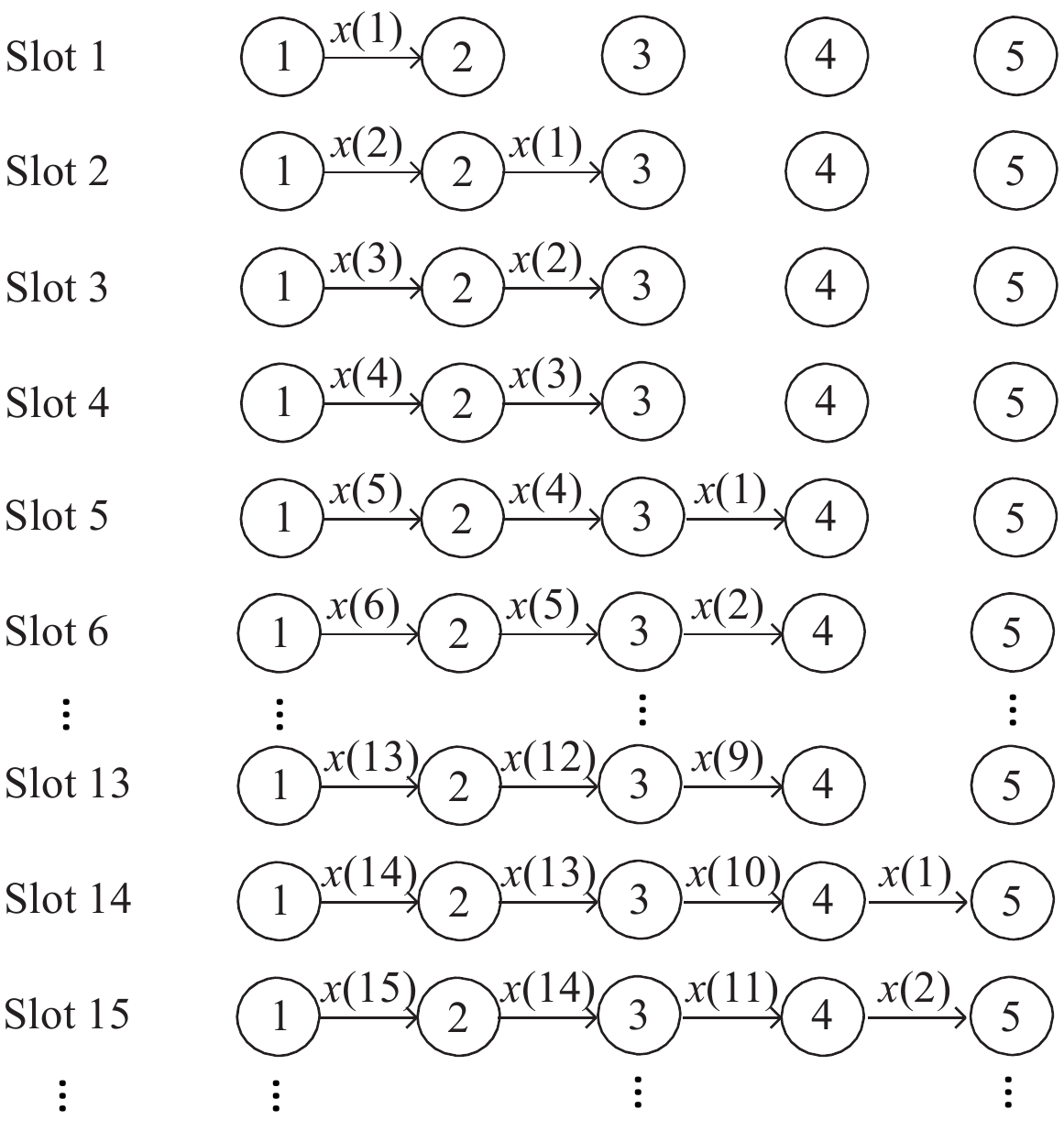}}

\caption{Example with $N=5$ showing the packet transmission procedure when
there exists multi-hop interference.}
\label{fig:constructiveExample}
\end{figure}

We first illustrate the interference cancellation procedure using
a constructive example with $N=5$, as shown in Fig. \ref{fig:constructiveExample}.
Same as in the existing E2E-KIC approach, node~$1$ transmits one
packet $x(t)$ in every slot $t\in\{1,2,...\}$. Because there is
no unknown multi-hop interference at the receiver of node~$2$, according
to (\ref{eq:systemModel}), the signal that node~$2$ receives in
slot $t\geq1$ (after cancelling any known interferences) is
\[
y_{2}(t)=h_{12}x\left(t\right)+z_{2}(t),
\]
from which it can successfully decode packet $x(t)$. 

Node~$2$ sends out its received packet in the next slot, so that
the received signal at node~$3$ in slot $t\geq2$ is
\[
y_{3}(t)=h_{13}x\left(t\right)+h_{23}x\left(t-1\right)+z_{3}(t),
\]
where $h_{23}x\left(t-1\right)$ carries the packet that node~$3$
intends to receive, and $h_{13}x\left(t\right)$ is the interference
signal\footnote{Obviously, if $h_{13}x\left(t\right)$ is strong enough, node~$3$
can receive the packet $x(t)$ sent by node~$1$ directly without
needing node~$2$ as relay. However, we assume here that $h_{13}x\left(t\right)$
is not strong enough for correct reception, because otherwise the
underlying routing mechanism will not choose node~$2$ as relay,
but $h_{13}x\left(t\right)$ may cause a remarkable level of interference
that affects the correct reception of $x\left(t-1\right)$ from $h_{23}x\left(t-1\right)$.}. 

Upon receiving $y_{3}(t)$, node~$3$ checks whether it can correctly
decode the packet. If not, it assumes that this is due to the unknown
multi-hop interference caused by node~$1$, in which case it waits
for signals received in additional slots. In slot $t+1$, node~$3$
receives
\[
y_{3}(t+1)=h_{13}x\left(t+1\right)+h_{23}x\left(t\right)+z_{3}(t+1),
\]
which contains a strong signal $h_{23}x\left(t\right)$ and a weak
signal $h_{13}x\left(t+1\right)$. 

Node~$3$ can use the strong signal in $y_{3}(t+1)$ that contains
$x(t)$ to eliminate the interference signal $h_{13}x\left(t\right)$
in $y_{3}(t)$. Explicitly, node~$3$ calculates the following:
\begin{align*}
 & g_{3,1}(t)\\
 & =y_{3}(t)-\frac{h_{13}}{h_{23}}y_{3}(t+1)\\
 & =h_{23}x\left(t-1\right)+z_{3}(t)-\frac{h_{13}}{h_{23}}\left(h_{13}x\left(t+1\right)+z_{3}(t+1)\right).
\end{align*}
The resulting signal now contains the useful signal $h_{23}x\left(t-1\right)$
and the interference-plus-noise signal $z_{3}(t)-\frac{h_{13}}{h_{23}}\left(h_{13}x\left(t+1\right)+z_{3}(t+1)\right)$.
We expect that the remaining interference-plus-noise signal in $g_{3,1}(t)$
is much smaller than $h_{13}x\left(t\right)+z_{3}(t)$ in $y_{3}(t)$,
because $\left|h_{13}\right|^{2}\ll\left|h_{23}\right|^{2}$ and the
additional noise signal $z_{3}(t+1)$ is usually much weaker than
the data-carrying signal. A rigorous analysis on the remaining interference-plus-noise
power will be given later in this paper. If the remaining interference-plus-noise
is still too strong, node~$3$ can wait for one more slot to receive
\[
y_{3}(t+2)=h_{13}x\left(t+2\right)+h_{23}x\left(t+1\right)+z_{3}(t+2)
\]
and perform another round of cancellation to cancel out the term involving
$x\left(t+1\right)$ in $g_{3,1}(t)$, yielding
\begin{align*}
g_{3,2}(t) & =g_{3,1}(t)+\frac{h_{13}^{2}}{h_{23}^{2}}y_{3}(t+2)\\
 & =h_{23}x\left(t-1\right)+z_{3}(t)-\frac{h_{13}}{h_{23}}z_{3}(t+1)\\
 & \quad+\frac{h_{13}^{2}}{h_{23}^{2}}\left(h_{13}x\left(t+2\right)+z_{3}(t+2)\right).
\end{align*}
Now, the remaining interference-plus-noise term is further smaller.
Suppose node~$3$ can now decode $x(t-1)$ from $g_{3,2}(t)$, it
can then send out $x(t-1)$ in the next slot $t+3$, so that node
$4$ receives the following signal for $t\geq5$:
\[
y_{4}(t)=h_{14}x\left(t\right)+h_{24}x\left(t-1\right)+h_{34}x\left(t-4\right)+z_{4}(t).
\]
Using the definition of $t_{j}$ and $\Delta_{j}$ in Section \ref{sec:System-Model},
we have $\Delta_{3}=t-t_{3}=4$, because node~$3$ transmits packet
$x(t-4)$ in slot $t$.

The operation is similar for node~$4$ and all remaining nodes. From
$y_{4}(t)$, node~$4$ intends to receive $x\left(t-4\right)$ transmitted
by node~$3$, and has to cancel out the signals $h_{14}x\left(t\right)$
and $h_{24}x\left(t-1\right)$ from nodes~$1$ and $2$, respectively.
Node~$4$ waits for four additional slots until it has received the
following signals:
\begin{align*}
y_{4}(t\!+\!3) & =h_{14}x\!\left(t\!+\!3\right)+h_{24}x\!\left(t\!+\!2\right)+h_{34}x\!\left(t\!-\!1\right)+z_{4}(t\!+\!3)\\
y_{4}(t\!+\!4) & =h_{14}x\left(t\!+\!4\right)+h_{24}x\left(t\!+\!3\right)+h_{34}x\left(t\right)+z_{4}(t\!+\!4).
\end{align*}
It then computes 
\[
g_{4,1}(t)=y_{4}(t)-\frac{h_{14}}{h_{34}}y_{4}(t+4)-\frac{h_{24}}{h_{34}}y_{4}(t+2)
\]
to cancel out $h_{14}x\left(t\right)$ and $h_{24}x\left(t-1\right)$
from $y_{4}(t)$, which introduces additional interference terms (of
much lower strength) involving $x(t+2)$, $x(t+3)$, and $x(t+4)$.
If needed, these can be cancelled out again using $y_{4}(t+6)$, $y_{4}(t+7)$,
and $y_{4}(t+8)$ received in subsequent slots. Assuming the SINR
is then sufficiently large for decoding $x\left(t-4\right)$, node~$4$
can send out $x\left(t-4\right)$ in slot $t+9$, and we have $\Delta_{4}=13$. 

In slots $t\geq14$, the destination node (indexed~$5$) receives
\[
y_{5}(t)\!=\!h_{15}x\!\left(t\right)\!+\!h_{25}x\!\left(t\!-\!1\right)\!+\!h_{35}x\!\left(t\!-\!4\right)\!+\!h_{45}x\!\left(t\!-\!13\right)\!+\!z_{5}(t).
\]
It intends to decode packet $x(t-13)$ which is carried in the strongest
signal component of $y_{5}(t)$. The remaining interference terms
can be cancelled using a similar approach as above.

\emph{A Note on Time Shifting:} Because we consider fixed values of
$h_{ji}$ and SINR threshold $\gamma$, the number of additional signals
(received after slot $t$) a node requires for successful interference
cancellation remains unchanged, thus the values of $\Delta_{j}$ are
also fixed\footnote{For more general cases with variable $h_{ji}$ and $\gamma$, our
scheme is still applicable. The only difference is that $\Delta_{j}$
may increase over time (or alternatively, we can fix it at a large
value at the beginning), and there may be additional idle slots at
certain nodes after packet transmission has started. We leave detailed
studies on this aspect for future work, and assume fixed $\Delta_{j}$
values in this paper.}. Therefore, from (\ref{eq:systemModel}), we have
\begin{equation}
y_{i}(t+\tau)=\sum_{j=1}^{i-1}h_{ji}\cdot x\left(t_{j}+\tau\right)+z_{i}(t+\tau)\label{eq:systemModelShifted}
\end{equation}
for integer values of $\tau$. This expression has also been used
while discussing the above example. 

Also due to fixed $\Delta_{j}$, \emph{each node is able to keep decoding
and sending new packets every slot }after it has successfully received
its first packet, as shown in Fig. \ref{fig:constructiveExample}.
This also holds for the general case presented next.

\subsection{General Case}

In general, when an arbitrary node~$i\geq3$ receives $y_{i}(t)$
(see (\ref{eq:systemModel})), we aim to cancel out all the interference
signals from nodes $1,2,...,i-2$, so that only $h_{(i-1),i}x\left(t_{i-1}\right)$
remains, from which we can decode $x\left(t_{i-1}\right)$. 

We define $\delta_{j}=t_{j}-t_{i-1}$ for $j<i-1$. According to (\ref{eq:systemModelShifted}),
we have
\begin{align*}
y_{i}(t+\delta_{j_{1}}) & =\sum_{j=1}^{i-1}h_{ji}x\left(t_{j}+\delta_{j_{1}}\right)+z_{i}(t+\delta_{j_{1}})\\
 & =h_{(i-1),i}x\left(t_{j_{1}}\right)\!+\!\sum_{j=1}^{i-2}h_{ji}x\left(t_{j}\!+\!\delta_{j_{1}}\right)\!+\!z_{i}(t\!+\!\delta_{j_{1}}),
\end{align*}
where we note that $t_{j_{1}}=t_{i-1}+\delta_{j_{1}}$. This expression
shows that node $i-1$ transmits $x\left(t_{j}\right)$ in slot $t+\delta_{j}$,
so we can use $y_{i}(t+\delta_{j})$ to cancel $h_{ji}x\left(t_{j}\right)$
in $y_{i}(t)$. 

After waiting for $\max_{j<i-1}\delta_{j}=\Delta_{i-1}$ slots, node~$i$
can cancel all the interferences from $y_{i}(t)$ by calculating
\begin{align}
g_{i,1}(t) & =y_{i}(t)-\sum_{j_{1}=1}^{i-2}\frac{h_{j_{1}i}}{h_{(i-1),i}}y_{i}(t+\delta_{j_{1}})\label{eq:gEqu1}\\
 & =h_{(i-1),i}x(t_{i-1})\!-\!\sum_{j_{1}=1}^{i-2}\sum_{j_{2}=1}^{i-2}\frac{h_{j_{1}i}h_{j_{2}i}}{h_{(i-1),i}}x(t_{i-1}\!+\!\delta_{j_{1}}\!+\!\delta_{j_{2}})\nonumber \\
 & \quad-\sum_{j_{1}=1}^{i-2}\frac{h_{j_{1}i}}{h_{(i-1),i}}z_{i}(t+\delta_{j_{1}})+z_{i}(t).\label{eq:gEqu2}
\end{align}
If another round of cancellation is needed to further reduce the remaining
interference term $\sum_{j_{1}=1}^{i-2}\sum_{j_{2}=1}^{i-2}\frac{h_{j_{1}i}h_{j_{2}i}}{h_{(i-1),i}}x(t_{i-1}+\delta_{j_{1}}+\delta_{j_{2}})$,
node~$i$ waits for additional $\Delta_{i-1}$ slots and calculates
\begin{align}
g_{i,2}(t) & =g_{i,1}(t)+\sum_{j_{1}=1}^{i-2}\sum_{j_{2}=1}^{i-2}\frac{h_{j_{1}i}h_{j_{2}i}}{h_{(i-1),i}^{2}}y_{i}(t+\delta_{j_{1}}+\delta_{j_{2}})\label{eq:gEqu3}\\
 & =h_{(i-1),i}\cdot x(t_{i-1})\nonumber \\
 & \quad+\!\sum_{j_{1}=1}^{i-2}\sum_{j_{2}=1}^{i-2}\sum_{j_{3}=1}^{i-2}\frac{h_{j_{1}i}h_{j_{2}i}h_{j_{3}i}}{h_{(i-1),i}^{2}}x(t_{i-1}\!+\!\delta_{j_{1}}\!+\!\delta_{j_{2}}\!+\!\delta_{j_{3}})\nonumber \\
 & \quad+\sum_{j_{1}=1}^{i-2}\sum_{j_{2}=1}^{i-2}\frac{h_{j_{1}i}h_{j_{2}i}}{h_{(i-1),i}^{2}}z_{i}(t+\delta_{j_{1}}+\delta_{j_{2}})\nonumber \\
 & \quad-\sum_{j_{1}=1}^{i-2}\frac{h_{j_{1}i}}{h_{(i-1),i}}z_{i}(t+\delta_{j_{1}})+z_{i}(t).\label{eq:gEqu4}
\end{align}

Considering the general case, the $m$-th round of cancellation is
performed according to (\ref{eq:mthIteration1}), which can be expanded
as (\ref{eq:mthIteration2}), where we denote $g_{i,0}(t)=y_{i}(t)$
for simplicity. These expressions are obtained by iteratively applying
the following \emph{cancellation principle}: 
\begin{itemize}
\item Cancel all interference terms from $g_{i,m-1}(t)$ in the $m$-th
cancellation round, without considering any additional interference
terms that are brought in by the cancellation process (these additional
interferences are cancelled in the next cancellation round if needed).
\end{itemize}
We note that the sign of interference signals inverses in every cancellation
round, thus the ``$\pm$'' sign in (\ref{eq:mthIteration1}) takes
``$+$'' when $m$ is even and ``$-$'' when $m$ is odd. In (\ref{eq:mthIteration1})
and (\ref{eq:mthIteration2}), ``$\sum_{j_{1}=1}^{i-2}\sum_{j_{2}=1}^{i-2}\cdots\sum_{j_{m-1}=1}^{i-2}\sum_{j_{m}=1}^{i-2}$''
stands for the concatenation of $m$ sums with their respective iteration
variables $j_{1},...,j_{m}$. It can be verified that (\ref{eq:gEqu1})--(\ref{eq:gEqu4})
are special cases of (\ref{eq:mthIteration1}) and (\ref{eq:mthIteration2}).

After $m$ rounds of cancellation, $m\Delta_{i-1}$ slots have elapsed
since slot $t$. If node~$i$ can successfully decode $x\left(t_{i-1}\right)$
now, it sends out $x\left(t_{i-1}\right)$ in slot $t+m\Delta_{i-1}+1$.

\subsubsection*{End-to-End Delay}

Assume that $m$ is the same for different nodes~$i$, we have $\Delta_{i}-\Delta_{i-1}=t_{i-1}-t_{i}=m\Delta_{i-1}+1$,
yielding
\begin{equation}
\Delta_{i}=(m+1)\Delta_{i-1}+1,\label{eq:diffDelta}
\end{equation}
which is a difference equation with respect to $\Delta_{i}$. Noting
that $\Delta_{1}=0$, we can solve the difference equation (\ref{eq:diffDelta})
as
\begin{equation}
\Delta_{i}=\begin{cases}
i-1 & \textrm{if }m=0\\
\frac{(m+1)^{i-1}-1}{m}, & \textrm{if }m>0
\end{cases}.\label{eq:diffDeltaSolution}
\end{equation}
We note that $\Delta_{i}$ is the \emph{end-to-end delay }(expressed
in the number of timeslots) of packet transmission from node~$1$
to node~$i$. Therefore, (\ref{eq:diffDeltaSolution}) is an expression
for calculating the delay. The total delay from node~$1$ to node~$N$
is $\Delta_{N}$.

\section{When Does It Work?}

In the following, we discuss insights behind the aforementioned interference
cancellation scheme, and derive conditions under which the proposed
scheme works.

\subsection{Upper Bound of Interference-Plus-Noise Power}

We note that in (\ref{eq:mthIteration2}), all the signal components
starting from the second term (i.e., excluding $h_{(i-1),i}x(t_{i-1})$)
are either interference or noise. Let $P_{I,m}$ denote the total
power of these interference-plus-noise signals in $g_{i,m}(t)$. We
have an upper bound of $P_{I,m}$. 

\begin{proposition} For $i\geq3$, $P_{I,m}$ has the following upper
bound: 
\begin{equation}
P_{I,m}\leq\left|h_{(i-2),i}\right|^{2}(i-2)\rho^{m}P_{T}+\sigma^{2}\frac{1-\rho^{m+1}}{1-\rho},\label{eq:PIBoundRho}
\end{equation}
where we recall that $P_{T}$ denotes the transmission power and $\sigma^{2}$
denotes the noise power, $\rho$ is defined as 
\begin{equation}
\rho=\frac{\left|h_{(i-2),i}\right|^{2}(i-2)}{\left|h_{(i-1),i}\right|^{2}}.\label{eq:rhoDef}
\end{equation}
\end{proposition} \begin{proof} Because $|h_{ji}|^{2}\leq\left|h_{(i-2),i}\right|^{2}$
for $j\in[1,i-2]$, by finding the power of $g_{i,m}(t)$ from its
expression (\ref{eq:mthIteration2}), replacing the coefficients $\left|h_{ji}\right|^{2}$
at the numerators with the upper bound $\left|h_{(i-2),i}\right|^{2}$,
and noting that the noise-related terms constitute a geometric series,
we have
\begin{align*}
P_{I,m} & \leq\sum_{j_{1}=1}^{i-2}\sum_{j_{2}=1}^{i-2}\cdots\sum_{j_{m+1}=1}^{i-2}\frac{\left|h_{(i-2),i}\right|^{2(m+1)}}{\left|h_{(i-1),i}\right|^{2m}}P_{T}\\
 & \quad+\sigma^{2}\sum_{\theta=0}^{m}\frac{\left|h_{(i-2),i}\right|^{2\theta}}{\left|h_{(i-1),i}\right|^{2\theta}}(i-2)^{\theta}\\
 & =\left|h_{(i-2),i}\right|^{2}(i-2)\left(\frac{\left|h_{(i-2),i}\right|^{2}(i-2)}{\left|h_{(i-1),i}\right|^{2}}\right)^{m}P_{T}\\
 & \quad+\sigma^{2}\sum_{\theta=0}^{m}\left(\frac{\left|h_{(i-2),i}\right|^{2}(i-2)}{\left|h_{(i-1),i}\right|^{2}}\right)^{\theta}\\
 & =\left|h_{(i-2),i}\right|^{2}(i-2)\rho^{m}P_{T}+\sigma^{2}\frac{1-\rho^{m+1}}{1-\rho}.
\end{align*}
\end{proof}

\subsection{Sufficient Condition for Successful Packet Reception}

\begin{proposition} \label{prop:cancellationRound} Node~$i\geq3$
can successfully decode $x(t_{i-1})$ from $g_{i,m}(t)$ when the
following conditions hold:
\begin{align}
i & <\frac{\left|h_{(i-1),i}\right|^{2}}{\left|h_{(i-2),i}\right|^{2}}-\frac{\gamma\sigma^{2}}{\left|h_{(i-2),i}\right|^{2}P_{T}}+2\label{eq:whenToWorkConditionSimplified1}\\
m & \geq\log_{\rho}\left(\frac{\frac{P_{T}}{\gamma}-\frac{\sigma^{2}}{\left|h_{(i-1),i}\right|^{2}-\left|h_{(i-2),i}\right|^{2}(i-2)}}{P_{T}-\frac{\sigma^{2}}{\left|h_{(i-1),i}\right|^{2}-\left|h_{(i-2),i}\right|^{2}(i-2)}}\right)-1.\label{eq:whenToWorkConditionSimplified2}
\end{align}
\end{proposition} \begin{proof} We note that node~$i$ can successfully
decode $x(t_{i-1})$ when the following condition holds:
\begin{equation}
\frac{\left|h_{(i-1),i}\right|^{2}P_{T}}{P_{I,m}}\geq\gamma.\label{eq:SINRTh}
\end{equation}
Combining (\ref{eq:SINRTh}) with (\ref{eq:PIBoundRho}), we obtain
a stricter condition for successful reception:
\begin{equation}
\frac{\left|h_{(i-1),i}\right|^{2}P_{T}}{\left|h_{(i-2),i}\right|^{2}(i-2)\rho^{m}P_{T}+\sigma^{2}\frac{1-\rho^{m+1}}{1-\rho}}\geq\gamma,\label{eq:ineqCancellationRoundProof1}
\end{equation}
which is equivalent to
\[
\frac{\left|h_{(i-1),i}\right|^{2}P_{T}}{\gamma}-\frac{\sigma^{2}}{1\!-\!\rho}\geq\left(\left|h_{(i-2),i}\right|^{2}(i\!-\!2)P_{T}-\frac{\sigma^{2}\rho}{1\!-\!\rho}\right)\rho^{m}.
\]
Substituting (\ref{eq:rhoDef}) into the above inequality, we know
that (\ref{eq:ineqCancellationRoundProof1}) is equivalent to
\begin{equation}
\frac{\left|h_{(i-1),i}\right|^{2}P_{T}}{\gamma}-\frac{\sigma^{2}}{1-\rho}\geq\left(\left|h_{(i-1),i}\right|^{2}P_{T}-\frac{\sigma^{2}}{1-\rho}\right)\rho^{m+1}.\label{eq:ineqCancellationRoundProof2}
\end{equation}

Using elementary algebra, we can easily verify that when the following
conditions are satisfied:
\begin{equation}
\rho<1,\label{eq:whenToWorkCondition1}
\end{equation}
\begin{equation}
\frac{\left|h_{(i-1),i}\right|^{2}P_{T}}{\gamma}-\frac{\sigma^{2}}{1-\rho}>0\label{eq:whenToWorkCondition2}
\end{equation}
(thus $\left|h_{(i-1),i}\right|^{2}P_{T}-\frac{\sigma^{2}}{1-\rho}>0$
because $\gamma\geq1$),
\begin{equation}
m\geq\log_{\rho}\left(\frac{\frac{\left|h_{(i-1),i}\right|^{2}P_{T}}{\gamma}-\frac{\sigma^{2}}{1-\rho}}{\left|h_{(i-1),i}\right|^{2}P_{T}-\frac{\sigma^{2}}{1-\rho}}\right)-1,\label{eq:whenToWorkCondition3}
\end{equation}
then (\ref{eq:ineqCancellationRoundProof2}) is satisfied, thus (\ref{eq:SINRTh})
and (\ref{eq:ineqCancellationRoundProof1}) are also satisfied. 

Substituting (\ref{eq:rhoDef}) into (\ref{eq:whenToWorkCondition1})--(\ref{eq:whenToWorkCondition3}),
we can show that (\ref{eq:whenToWorkCondition1}) is equivalent to
$i<\frac{\left|h_{(i-1),i}\right|^{2}}{\left|h_{(i-2),i}\right|^{2}}+2$,
and (\ref{eq:whenToWorkCondition2}) is equivalent to (\ref{eq:whenToWorkConditionSimplified1}).
Hence, when (\ref{eq:whenToWorkConditionSimplified1}) holds, (\ref{eq:whenToWorkCondition1})
and (\ref{eq:whenToWorkCondition2}) also hold. The right-hand side
(RHS) of (\ref{eq:whenToWorkCondition3}) is equal to the RHS of (\ref{eq:whenToWorkConditionSimplified2}),
thus (\ref{eq:whenToWorkCondition3}) holds when (\ref{eq:whenToWorkConditionSimplified2})
holds. \end{proof}

Proposition~\ref{prop:cancellationRound}, which gives a set of sufficient
conditions, provides a possibly conservative value on the number of
required cancellation rounds $m$, governed by the lower bound in
(\ref{eq:whenToWorkConditionSimplified2}). If (\ref{eq:whenToWorkConditionSimplified1})
is satisfied, packet $x(t_{i-1})$ can be successfully received by
node~$i$ when $m$ satisfies (\ref{eq:whenToWorkConditionSimplified2}).
We note that the minimum value of $m$ satisfying (\ref{eq:whenToWorkConditionSimplified2})
is always finite if (\ref{eq:whenToWorkConditionSimplified1}) (thus
(\ref{eq:whenToWorkCondition1}) and (\ref{eq:whenToWorkCondition2})
in the proof) holds. Therefore, we can say that multi-hop interference
can be effectively cancelled within a finite number of rounds when
(\ref{eq:whenToWorkConditionSimplified1}) holds.

\subsection{How Many Nodes Are Allowed? \label{sub:How-Many-Nodes}}

Condition (\ref{eq:whenToWorkConditionSimplified1}) imposes an upper
bound on $i$, implying that $i$ (thus $N$) cannot be too large.
Assume that $\left|h_{ji}\right|^{2}\propto\frac{1}{d_{ji}^{\alpha}}$,
where $d_{ji}$ is the geographical distance between nodes $j$ and
$i$, and $\alpha$ is the path-loss exponent related to the wireless
environment. Suppose we have equally spaced nodes, such that $d_{(i-2),i}=2d_{(i-1),i}$.
Define a quantity $B>0$ such that 
\begin{equation}
\frac{\left|h_{(i-2),i}\right|^{2}P_{T}}{\sigma^{2}}=\frac{\gamma}{B},\label{eq:BDef}
\end{equation}
thus $\frac{\gamma\sigma^{2}}{\left|h_{(i-2),i}\right|^{2}P_{T}}=B$.
When $B=2-\epsilon$ for an arbitrarily small $\epsilon>0$, condition
(\ref{eq:whenToWorkConditionSimplified1}) becomes
\begin{equation}
i<2^{\alpha}+\epsilon.
\end{equation}
In practical environments, we normally have $2<\alpha<6$, which means
that the proposed approach always works for chain networks with $N=4$
nodes. Whether it works with more nodes depends on the value of $\alpha$
in the environment. For example when $\alpha=4$, we can have $N=16$.
Also note that we are using the sufficient condition given by Proposition~\ref{prop:cancellationRound},
so our results here may be conservative (as we will see in the numerical
results next section), and the proposed approach may work well with
much larger values of $N$.

The value of $B$ in (\ref{eq:BDef}) is related to the placement
of nodes. A smaller value of $B$ implies a larger $\left|h_{(i-2),i}\right|^{2}$.
We should normally have $B>1$ because otherwise nodes $i-2$ and
$i$ can communicate directly (when there is no interference) and
node $i-1$ may be ignored by common routing protocols. Aspects related
to optimal node placement is open for future research. Because we
fix the ratio $\nicefrac{\left|h_{(i-1),i}\right|^{2}}{\left|h_{(i-2),i}\right|^{2}}=2^{\alpha}$
here, a smaller $B$ also yields a larger $\left|h_{(i-1),i}\right|^{2}$.

\section{Numerical Results}

\begin{figure}
\center{\subfigure[]{\includegraphics[width=0.9\linewidth]{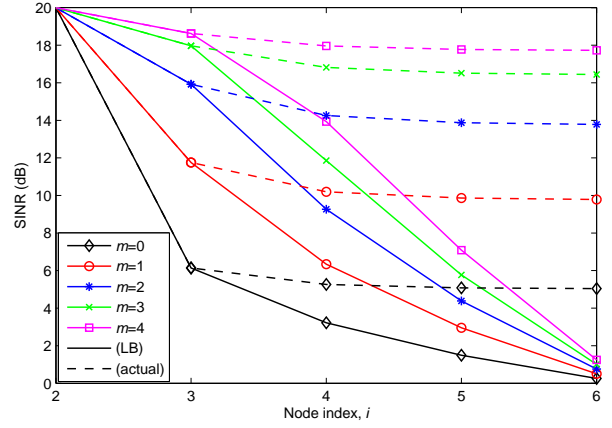}} \\ \subfigure[]{\includegraphics[width=0.9\linewidth]{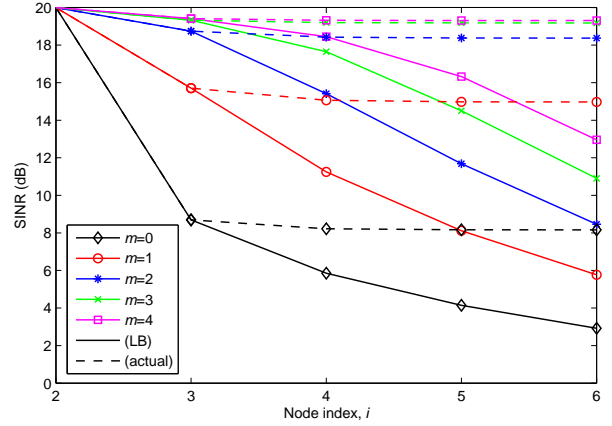}} \\ \subfigure[]{
\includegraphics[width=0.9\linewidth]{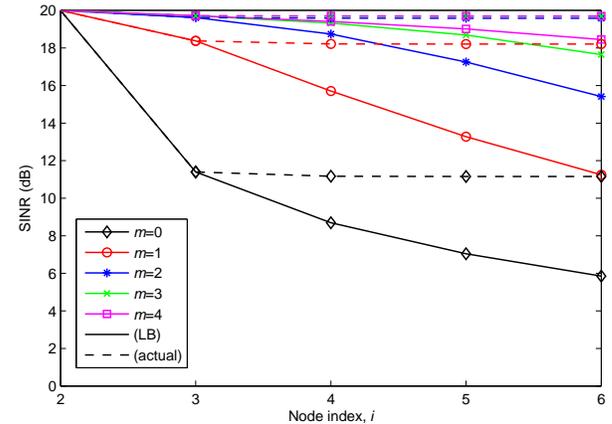}}}

\caption{SINR at node~$i$, under different values of $m$, where (a) $\alpha=2.1$,
(b) $\alpha=3$, (c) $\alpha=4$. The solid lines denote the theoretical
lower bound (LB) and the dashed lines denote the actual values.}
\label{fig:simSINR}
\end{figure}

\begin{figure}
\center{\includegraphics[width=0.7\linewidth]{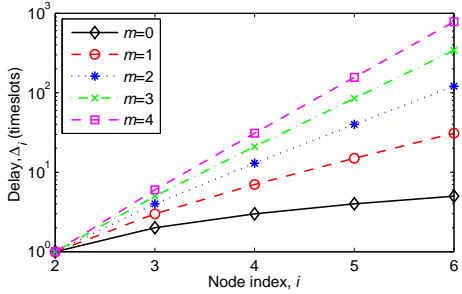}}

\caption{End-to-end delay from node~$1$ to node~$i$, under different values
of~$m$.}
\label{fig:simDelay}
\end{figure}

We now present some numerical results of the proposed interference
cancellation scheme based on the above analysis. 

We first evaluate the SINR after interference cancellation, at different
nodes and with different number of cancellation rounds $m$. The same
static channel model as in Section \ref{sub:How-Many-Nodes} is used
in the evaluation. More realistic cases involving random channel coefficients
is left for future work. We set the signal-to-noise (SNR) ratio of
single-hop transmission (without interference) to $20$~dB, based
on which we can calculate the SINR of $g_{i,m}(t)$. We use two different
methods to calculate the SINR. One is based on the theoretical interference-plus-noise
upper bound in (\ref{eq:PIBoundRho}), which gives an SINR lower bound
(LB); the other is based on the actual signal expression in (\ref{eq:mthIteration2}).
Obviously, the second method gives the actual SINR, but it is also
more complex than the first method. Fig.~\ref{fig:simSINR} shows
the results with different path-loss exponents $\alpha$.

Note that $m=0$ corresponds to the existing E2E-KIC approach without
cancelling unknown multi-hop interferences. We can see from Fig.~\ref{fig:simSINR}
that the proposed approach (with $m>0$) provides a significant gain
in SINR. Even when $m=1$, the gain compared to $m=0$ is over $4$~dB
(based on the actual SINR values) in all cases we evaluate here. The
gain goes higher when $\alpha$ is larger because the multi-hop interference
is more attenuated with a larger $\alpha$. This SINR gain is important
for E2E-KIC to work in scenarios with multi-hop interference. For
example, consider $\alpha=3$ and $\gamma=10$~dB, E2E-KIC would
not work without the proposed approach (i.e., $m=0$) when $N\geq3$,
because the SINR is below the $\gamma=10$~dB threshold; using a
single round of cancellation with the proposed approach (i.e., $m=1$),
the SINR becomes significantly above $\gamma=10$~dB and E2E-KIC
works (at least for $N\in\{3,4,5,6\}$). 

At node $i=2$, the SINR is equal to the single-hop SNR ($20$~dB),
because there is no unknown multi-hop interference at this node. For
nodes $i\geq3$, the SINR is below $20$~dB and depends on the value
of $m$. As expected, a larger $m$ yields a higher SINR. When $m$
is large enough, the SINR becomes very close to the single-hop SNR. 

The theoretical LB is equal to the actual SINR when $i=3$. At $i>3$,
the LB tends to significantly underestimate the SINR and thus the
performance of the proposed approach. This is because the interference-plus-noise
upper bound in (\ref{eq:PIBoundRho}) replaces all $\left|h_{ji}\right|^{2}$
for $j<i-2$ with $\left|h_{(i-2),i}\right|^{2}$, which remarkably
enlarges the interferences from nodes that are three or more hops
away. This validates that the results presented in Section \ref{sub:How-Many-Nodes}
are quite conservative, and the proposed approach should actually
perform much better. Since it is hard to obtain meaningful analytical
results directly from (\ref{eq:mthIteration2}), an interesting direction
for future work is to find tighter theoretical bounds for the interference-plus-noise
power.

Fig. \ref{fig:simDelay} shows the end-to-end delay from node~$1$
to node $i$, i.e., the $\Delta_{i}$ values computed from (\ref{eq:diffDeltaSolution}).
When $i=N$, this corresponds to the total delay from the source node
(node~$1$) to the destination node (node~$N$). We see that the
delay increases with the number of cancellation rounds $m$ and the
node index~$i$. However, a good news is that $m$ usually does not
need to be too large to obtain a reasonably high SINR gain, as seen
in Fig.~\ref{fig:simSINR}. Also note that the ``packets '' in this
paper \emph{do not} need to be full data packets on the network layer.
In fact, they can be as short as a single symbol if a proper scheduling
and synchronization scheme is available. 

Furthermore, the delay is due to the idling time before each node
receives its first packet. Once the first packet has been received,
there is no additional delay for receiving subsequent packets. A large
chain network does not necessarily have all nodes waiting for the
first packet to arrive (or waiting for the whole transmission to complete).
It can schedule some other transmissions in those idling slots instead.
With such a properly designed scheduling mechanism, the throughput
benefit of E2E-KIC can still be maintained although the transmission
delay may be large.

\section{Conclusion}

In this paper, we have studied E2E-KIC with multi-hop interference.
We have noted that among all multi-hop interferences, the unknown
interferences from nodes preceding the current node are the most challenging,
while the other known interferences can be cancelled with existing
KIC approaches. We have therefore proposed an approach that iteratively
reduces the strength of unknown interferences to a sufficiently low
level, so that nodes can successfully receive new packets. Analytical
and numerical results confirm the effectiveness and also provide insights
of the proposed approach.

\bibliographystyle{IEEEtran}
\bibliography{IEEEabrv,E2E-KIC_interference_analysis}

\end{document}